\newif\ifConferenceVersion
\ConferenceVersionfalse

\documentclass[]{article}
\usepackage{fancyhdr}
\usepackage{times}

\usepackage{nicefrac}

\usepackage{diagbox}
\usepackage{multirow}

\usepackage[utf8]{inputenc}
\usepackage[american]{babel}
\usepackage{amssymb}
\usepackage{amsmath}
\usepackage{amsthm}
\usepackage{listings}
\usepackage{graphicx}
\usepackage{color}
\usepackage{subfig}
\usepackage{xspace}
\usepackage[colorinlistoftodos]{todonotes}
\usepackage{enumerate}
\usepackage{xargs}
\usepackage{mathtools}  
\usepackage[ruled,vlined,linesnumbered]{algorithm2e}  
\makeatletter
\def\@IEEEsectpunct{} 
\def\paragraph{\@startsection{paragraph}{4}{\z@}{-3.25ex \@plus-1ex \@minus-.2ex}%
	{-1em}{\normalfont\normalsize\sffamily\bfseries}}
\makeatother

\makeatletter
\renewcommand*{\@fnsymbol}[1]{}
\makeatother

\let\originalleft\left
\let\originalright\right
\renewcommand{\left}{\mathopen{}\mathclose\bgroup\originalleft}
\renewcommand{\right}{\aftergroup\egroup\originalright}
\newcommandx*{\LDAUOmicron}[2][1=@pkling_false]{\mathrm{O}\ifthenelse{\equal{#1}{small}}{\bigl(#2\bigr)}{\left(#2\right)}}
\newcommandx*{\LDAUomicron}[2][1=@pkling_false]{\mathrm{o}\ifthenelse{\equal{#1}{small}}{\bigl(#2\bigr)}{\left(#2\right)}}
\newcommandx*{\LDAUOmega}[2][1=@pkling_false]{\Omega\ifthenelse{\equal{#1}{small}}{\bigl(#2\bigr)}{\left(#2\right)}}
\newcommandx*{\LDAUomega}[2][1=@pkling_false]{\omega\ifthenelse{\equal{#1}{small}}{\bigl(#2\bigr)}{\left(#2\right)}}
\newcommandx*{\LDAUTheta}[2][1=@pkling_false]{\Theta\ifthenelse{\equal{#1}{small}}{\bigl(#2\bigr)}{\left(#2\right)}}

\newtheorem{theorem}{Theorem}[section]
\newtheorem{lemma}[theorem]{Lemma}
\newtheorem{corollary}[theorem]{Corollary}
\newtheorem{definition}[theorem]{Definition}

\newtheorem{proposition}[theorem]{Proposition}

\newtheorem{observation}[theorem]{Observation}

\newcommand{\bigO}{\mathcal{O}}
\newcommand*{\N}{\mathbb{N}}

\newcommand{\MIN}{\textsc{Min}}
\newcommand{\MAX}{\textsc{Max}}

\newcommand{\Subprotocol}{\textsc{SubProtocol}\xspace}
\newcommand{\Dense}{\textsc{DenseProtocol}\xspace}
\newcommand{\Scattered}{\textsc{Top-K-Protocol}\xspace}

\newcommand{\etopk}{$\varepsilon$-Top-$k$-Position Monitoring\xspace}

\author{Alexander M\"acker \and Manuel Malatyali\thanks{This work was partially supported by the German Research Foundation (DFG) within the Priority Program ``Algorithms for Big Data'' (SPP 1736) and by the EU within FET project MULTIPLEX under contract no.\ 317532.} \and Friedhelm Meyer auf der Heide \\ [0.4em]
	Heinz Nixdorf Institute \& 	Computer Science Department\\
	Paderborn University, Germany\\[0.2em]
	\{amaecker, malatya, fmadh\}@hni.upb.de}
\date{}

\begin{document}
\title{On Competitive Algorithms for Approximations of Top-$k$-Position Monitoring of Distributed Streams}
	
\maketitle

\begin{abstract}
Consider the continuous distributed monitoring model in which $n$ distributed nodes, receiving individual data streams, are connected to a designated server.
The server is asked to continuously monitor a function defined over the values observed across all streams while minimizing the communication.
We study a variant in which the server is equipped with a broadcast channel and is supposed to keep track of an approximation of the set of nodes currently observing the $k$ largest values.
Such an approximate set is exact except for some imprecision in an $\varepsilon$-neighborhood of the $k$-th largest value.
This approximation of the Top-$k$-Position Monitoring Problem is of interest in cases where marginal changes (e.g.\ due to noise) in observed values can be ignored so that monitoring an approximation is sufficient and can reduce communication.

This paper extends our results from [6], where we have developed a filter-based online algorithm for the (exact) Top-k-Position Monitoring Problem. 
There we have presented a competitive analysis of our algorithm against an offline adversary that also is restricted to filter-based algorithms. 
Our new algorithms as well as their analyses use new methods. 
We analyze their competitiveness against adversaries that use both exact and approximate filter-based algorithms, and observe severe differences between the respective powers of these adversaries. 
\end{abstract}

\section{Introduction}

We consider a setting in which $n$ distributed nodes are connected to a central server.
Each node continuously observes a data stream and the server is asked to keep track of the value of some function defined over all streams.
In order to fulfill this task, nodes can communicate to the server, while the server can employ a broadcast channel to send a message to all nodes.

In an earlier paper~\cite{mmm}, we introduced and studied a problem called Top-$k$-Position Monitoring in which, at any time $t$, the server is interested in monitoring the $k$ nodes that are observing the largest values at this particular time $t$.
As a motivating example, picture a scenario in which a central load balancer within a local cluster of webservers is interested in keeping track of those nodes which are facing the highest loads.
We proposed an algorithm based on the notion of filters and analyzed its competitiveness with respect to an optimal filter-based offline algorithm. 
Filters are assigned by the server and are used as a means to indicate the nodes when they can resign to send updates; this particularly reduces communication when observed values are ``similar'' to the values observed in the previous time steps.

In this paper, we broaden the problem and investigate the monitoring of an approximation of the Top-$k$-Positions.
We study the problem of $\varepsilon$-Top-$k$-Position Monitoring, in which the server is supposed to maintain a subset of $k$ nodes such that all nodes observing ``clearly larger'' values than the node which observed the $k$-th largest value are within this set and no node observing a ``clearly smaller'' value belongs to this set.
Here, smaller/larger is meant to be understood with respect to $\varepsilon$ and the $k$-th largest value observed.
A detailed definition is given in Sect.~\ref{sec:model}.
Relaxing the problem in this direction can reduce communication while, in many cases, marginal or insignificant changes (e.g.\ due to noise) in observed values can be ignored and justify the sufficiency of an approximation.
Examples are situations where lots of nodes observe values oscillating around the $k$-th largest value and where this observation is not of any qualitative relevance for the server.
We design and analyze algorithms for $\varepsilon$-Top-$k$-Position Monitoring and, although we use these very tools of filters and competitive analysis \cite{mmm}, the imprecision/approximation requires fundamentally different online strategies for defining filters in order to obtain efficient solutions.

\subsection{Our Contribution}
In this paper we investigate a class of algorithms that are based on using filters and study their efficiency in terms of competitive analysis.

As a first technical contribution we analyze an algorithm (Sect.~\ref{sec:existence}) which allows the server to decide the logical disjunction of the (binary) values observed by the distributed nodes.
It uses a logarithmic number of rounds and a constant number of messages on expectation.
As a by-product, using this algorithm, the result on the competitiveness of the filter-based online algorithm in \cite{mmm} can be reduced from $\bigO(k \log n + \log \Delta \log n)$ to $\bigO(k \log n + \log \Delta)$,
for observed values from $\{0,1, \ldots, \Delta\}$.

Second, we also propose an online algorithm (Sect.~\ref{sec:scatteredData}) that is allowed to introduce an error of $\varepsilon \in (0, 1/2]$ in the output and compare it to an offline algorithm that solves the exact Top-$k$-Position Monitoring problem.
We show that this algorithm is $\bigO(k \log n + \log \log \Delta + \log \frac{1}{\varepsilon})$-competitive. 
Note that this imprecision allows to bring the $\log \Delta$ in the upper bound down to $\log \log \Delta$ for any constant $\varepsilon$.

We also investigate the setting in which also the offline algorithm is allowed to have an error in the output (Sect.~\ref{sec:denseData}).
We first show that these results are not comparable to previous results; we prove a lower bound on the competitiveness of $\Omega(n / k)$.
Our third and main technical contribution is an algorithm with a competitiveness of $\bigO(n^2 \log (\varepsilon \Delta) + n \log^2 (\varepsilon \Delta) + \log \log \Delta + \log \frac{1}{\varepsilon} )$ if the online and the offline algorithm may use an error of $\varepsilon$.

However, if we slightly decrease the allowed error for the offline algorithm, the lower bound on the competitiveness of $\Omega(n / k)$ still holds, while the upper bound is reduced to $\bigO(n + k \log n + \log \log \Delta + \log \frac{1}{\varepsilon})$.

\subsection{Related Work}
\label{sec:relatedWork}
Efficient computation of functions on big datasets in terms of streams has turned out to be an important topic of research with applications in network traffic analysis, text mining or databases (e.g.\ \cite{sanders} and \cite{muthu}). 

The \emph{Continuous Monitoring Model}, which we consider in this paper, was introduced by Cormode et al.\ \cite{cormodeSurvey} to model systems comprised of a server and $n$ nodes observing \emph{distributed} data streams.
The primary goal addressed within this model is the continuous computation of a function depending on the information available across all $n$ data streams up to the current time at a dedicated server. 
Subject to this main concern, the minimization of the overall number of messages exchanged between the nodes and the server usually determines the efficiency of a streaming algorithm. 
We refer to this model and enhance it by a broadcast channel as proposed by Cormode et al.\ in~\cite{cormodeFunction}. 

An important class of problems investigated in literature are threshold computations where the server is supposed to decide whether the current function value has reached some given threshold $\tau$. 
For monotone functions such as monitoring the number of distinct values or the sum over all values, exact characterizations in the deterministic case are known~\cite{cormodeSurvey,cormodeFunction}. 
However, non-monotone functions, e.g., the entropy~\cite{chakrabarti}, turned out to be much more complex to handle.

A general approach to reduce the communication when monitoring distrib\-uted streams is proposed in \cite{zhangFilter}. 
Zhang et al.\ introduce the notion of \emph{filters}, which are also an integral part of our algorithms.
They consider the problem of continuous skyline maintenance, in which a server is supposed to continuously maintain the skyline of dynamic objects. 
As they aim at minimizing the communication overhead between the server and the objects, they use a filter method that helps in avoiding the transmission of updates in case these updates cannot influence the skyline. 
More precisely, the objects are points of a $d$-dimensional space and filters are hyper-rectangles assigned by the server to the objects such that  as long as these points are within the assigned hyper-rectangle, updates need not be communicated to the server. 

Despite its online nature, by now streaming algorithms are barely studied in terms of competitiveness.
In their work \cite{yi}, Yi and Zhang were the first to study streaming algorithms with respect to their competitiveness and recently this approach was also applied in a few papers (\cite{lam, tang, mmm, giannakopoulos}).
In their model \cite{yi}, there is one node and one server and the goal is to keep the server informed about the current value of a function $f: \mathbb{Z}^+ \to \mathbb{Z}^d$ that is observed by the node and changes its value over time, while minimizing the number of messages. 
Yi and Zhang present an algorithm that is $\bigO(d^2 \log(d \cdot \delta))$-competitive if the last value received by the server might deviate by $\delta$ from the current value of $f$.
Recently, Tang et al.\ \cite{tang} extended this work by Yi and Zhang for the two-party setting to the distributed case.
They consider a model in which the server is supposed to track the current value of a (one-dimensional) function that is defined over a set of $n$ functions observed at the distributed nodes.
Among other things, they propose an algorithm for the case of a tree-topology in which the distributed nodes are the leaves of a tree connecting them to the server. 
They show that on any instance $I$ their algorithm incurs communication cost that is by a factor of $\bigO(h_{max}\log \delta)$, where $h_{max}$ represents the maximimum length of a path in the tree, larger than those of the best solution obtained by an online algorithm on $I$.

Following the idea of studying competitive algorithms for monitoring streams and the notion of filters, Lam et al.\ \cite{lam} present an algorithm for online dominance tracking of distributed streams.
In this problem a server always has to be informed about the dominance relationship between $n$ distributed nodes each observing an online stream of $d$-dimensional values. 
Their algorithm is based on the idea of filters and they show that a mid-point strategy, which sets filters to be the mid-point between neighboring nodes, is $\bigO(d \log U)$-competitive with respect to the number of messages sent in comparison to an offline algorithm that sets filters optimally. 

While we loosely motivated our search for approximate solutions by noise in the introduction, in other problems noise is a major concern and explicitly addressed.
For example, consider streaming algorithms for estimating statistical parameters like frequency moments \cite{zhangNoise}.
In such problems, certain elements from the universe may appear in different forms due to noise and thus, should actually be treated as the same element.

\section{Preliminaries}
\label{sec:model}
In our setting there are $n$ distributed nodes $\{1, \ldots, n\}$.
Each node $i$ receives a continuous data stream $(v_{i}^1, v_i^2, v_i^3 \ldots)$, which can be exclusively observed by node $i$.
At time $t$, $v_i^t \in \mathbb{N}$ is observed and no $v_i^{t'}$, $t' > t$, is known.
We omit the index $t$ if it is clear from the context. 

Following the model in \cite{cormodeFunction}, we allow that between any two consecutive time steps, a \emph{communication protocol} exchanging messages between the server and the nodes may take place.
The communication protocol is allowed to use an amount of rounds which is polylogarithmic in $n$ and $\max_{1\leq i\leq n}(v_i^t)$.
The nodes can communicate to the server while the server can communicate to single nodes or utilize a broadcast channel to communicate a message that is received by all nodes at the same time.
These communication methods incur unit communication cost per message, we assume instant delivery, and a message at time $t$ is allowed to have a size at most logarithmic in $n$ and $\max_{1\leq i\leq n}(v_i^t)$. 

\paragraph*{Problem Description} 
Consider the Top-$k$-Position Monitoring problem \cite{mmm}, in which the server is asked to keep track of the set of nodes currently holding the $k$ largest values.
We relax this definition and study an approximate variant of the problem in which this set is exact except for nodes in a small neighborhood around the $k$-th largest value. 
We denote by $\pi(k, t)$ the node which observes the $k$-th largest value at time $t$ and denote by top-$k \coloneqq \{ i \in \{1, \ldots, k\} : \pi (i, t) \}$ the nodes observing the $k$ largest values.
Given an error $0<\varepsilon<1$, for a time $t$ we denote by $E(t) \coloneqq (\frac{1}{1-\varepsilon} v_{\pi(k, t)}^t, \infty]$ the range of values that are clearly larger than the $k$-th largest value and by $A(t) \coloneqq [(1-\varepsilon) v_{\pi(k, t)}^t, \frac{1}{1-\varepsilon} v_{\pi(k, t)}^t]$ the \emph{$\varepsilon$-neighborhood} around the $k$-th largest value.
Furthermore, we denote by $\mathcal{K}(t) \coloneqq \{ i : v_i^t \in A(t)\}$ the nodes in the $\varepsilon$-neighborhood around the $k$-th largest value.
Then, at any time $t$, the server is supposed to know the nodes $\mathcal{F}(t) = \mathcal{F}_E(t)\: \cup \: \mathcal{F}_A(t) =  \{i_1, \ldots, i_k\}$ according to the following properties:

\begin{enumerate}
	\item $\mathcal{F}_E(t) = \{i : v_i^t \in E(t)\}$ and

	\item $\mathcal{F}_A(t) \subseteq \mathcal{K}(t) = \{ i : v_i^t \in A(t) \} $, such that $| \mathcal{F}_A(t)| = k - |\mathcal{F}_E(t)|$ holds.
	
\end{enumerate}
Denote by $\Delta$ the maximal value observed by some node (which may not be known beforehand).
We use $\mathcal{F}_1 = \mathcal{F}(t)$ if $t$ is clear from the context, $\mathcal{F}_2 = \{1,  \ldots, n\} \setminus \mathcal{F}(t)$, and call $\mathcal{F}^*$ the \emph{output} of an optimal offline algorithm.
If the $k$-th and the $(k+1)$-st largest value differ by more than $\varepsilon \, v_{\pi(k, t)}^t$, $\mathcal{F}(t)$ coincides with the set in the (exact) Top-$k$-Position Monitoring problem and hence, $\mathcal{F}(t)$ is unique.
We denote by $\sigma(t)  \coloneqq | \mathcal{K}(t) |$ the number of nodes at time $t$ which are in the $\varepsilon$-neighborhood of the $k$-th largest value and $\sigma \coloneqq \max_t \sigma(t)$.
Note that $|\mathcal{K}(t)| = 1$ implies that $\mathcal{F}(t)$ is unique.
Furthermore for solving the exact Top-$k$-Position Monitoring problem we assume that the values are distinct (at least by using the nodes' identifiers to break ties in case the same value is observed by several nodes).

\subsection{Filter-Based Algorithms \& Competitive Analysis}
A set of filters is a collection of intervals, one assigned to each node, such that as long as the observed values at each node are within its respective interval, the output $\mathcal{F}(t)$ need not change. 
For the problem at hand, this general idea of filters translates to the following definition.

\begin{definition}\cite{mmm}
\label{def:filters}
	For a fixed time $t$, a \emph{set of filters} is defined as an $n$-tuple of intervals $(F_1^t, \ldots, F_n^t)$, $F_i \subseteq \mathbb{N} \cup \{\infty \}$ and $v_i \in F_i$, 
	such that as long as the value of node $i$ only changes within its interval (i.e.\ $v_i \in F_i$), the value of the output $\mathcal{F}$ need not change.
\end{definition}

Observe that each pair of filters $\left(F_i,F_j\right)$ of nodes $i \in \mathcal{F}(t)$ and $j \notin \mathcal{F}(t)$ must be disjoint except for a small overlapping. 
This observation can be stated formally as follows.
\begin{observation}
	\label{le:apxFilterDef}
	For a fixed time $t$, an $n$-tuple of intervals is a set of filters if and only if for all pairs $i \in \mathcal{F}(t)$ and $j \notin \mathcal{F}(t)$ the following holds: $v_i \in F_i = [\ell_i, u_i]$, $v_j \in F_j = [\ell_j, u_j]$ and $\ell_i \geq (1-\varepsilon)u_j$.
\end{observation}

In our model, we assume that nodes are assigned such filters by the server. 
If a node observes a value that is larger than the upper bound of its filter, we say the node \emph{violates its filter from below}. 
A violation \emph{from above} is defined analogously.
If such a violation occurs, the node may report it and its current value to the server.
In contrast to \cite{mmm}, we allow the server to assign ``invalid'' filters, i.e., there are affected nodes that directly observe a filter-violation.
However, for such an algorithm to be correct, we demand that the intervals assigned to the nodes at the end of the protocol at time $t$ and thus, before observations at time $t+1$, constitute a (valid) set of filters. 
We call such an algorithm \emph{filter-based}.
Note that the fact that we allow invalid filters (in contrast to \cite{mmm}) simplifies the presentation of the algorithms in the following. 
However, using a constant overhead the protocols can be changed such that only (valid) filters are sent to the nodes. 

\paragraph*{Competitiveness}
To analyze the quality of our online algorithms, we use analysis based on competitiveness and compare the communication induced by the algorithms to that of an adversary's offline algorithm.

Similar to \cite{lam} and \cite{mmm}, we consider adversaries that are restricted to use filter-based offline algorithms and hence, OPT is lower bounded by the number of filter updates.
However, we compare our algorithms against several adversaries which differ in terms of whether their offline algorithm solves the exact Top-$k$-Position Monitoring Problem or \etopk.
The adversaries are assumed to be adaptive, i.e., values observed by a node are given by an adversary who knows the algorithm's code, the current state of each node and the server and the results of random experiments.

An online algorithm is said to have a competitiveness of $c$ if the number of messages is at most by a factor of $c$ larger than that of the adversary's offline algorithm.

\subsection{Observations and Lemmas}

Define for some fixed set $\mathcal{S} \subseteq \{1, \ldots, n\}$ the minimum of the values observed by nodes in $\mathcal{S}$ during a time period $[t, t']$ as $\MIN_\mathcal{S}(t, t')$ and the maximum of the values observed during the same period as $\MAX_\mathcal{S}(t, t')$.

\begin{definition}
\label{def:maxMin}
Let $t, t'$ be given times with $t' \geq t$.
For a subset of nodes $\mathcal{S} \subseteq \{1, \ldots, n\}$ the values $\MAX_\mathcal{S}(t, t') \coloneqq \max_{t \leq t^* \leq t'} \max_{i \in \mathcal{S}}  ( v_i^{t^*} )$ and $\MIN_\mathcal{S}(t, t')$ are defined analogously. 
\end{definition}
 
Observe that it is sufficient for an optimal offline algorithm to only make use of two different filters $F_1$ and $F_2$.

\begin{proposition}
\label{ob:twofilters}
Without loss of generality, we may assume that an optimal offline algorithm only uses two different filters at any time. 
\end{proposition}
\begin{proof}
Let $[t, t']$ be an interval during which $OPT$ does not communicate.
We fix its output $\mathcal{F}^*_1$ and define $\mathcal{F}^*_2 \coloneqq \{1, \ldots, n\} \setminus \mathcal{F}^*_1$.
If $OPT$ only uses two different filters throughout the interval, we are done. 
Otherwise, using $\mathcal{F}^*_1$ as output throughout the interval $[t, t']$ and filters $F_1 = [\textsc{Min}_{\mathcal{F}^*_1}(t,t'), \infty]$ and $F_2 = [0, \textsc{Max}_{\mathcal{F}^*_2}(t,t')]$, which must be feasible due to the assumption that $OPT$ originally assigned filters that lead to no communication, leads to no communication within the considered interval.
\end{proof}

The following lemma generalizes a lemma in \cite{mmm} to \etopk.
Assuming the optimal offline algorithm did not change the set of filters during a time period $[t, t']$, the minimum value observed by nodes in $\mathcal{F}^*_1$ can only be slightly smaller than the maximum value observed by nodes in $\mathcal{F}^*_2$.

\begin{lemma}
\label{le:filterChange}
If $OPT$ uses the same set of filters $F_1, F_2$ during $[t,t']$, 
then it holds $\MIN_{\mathcal{F}^*_1}(t, t') \geq (1-\varepsilon)\ \MAX_{\mathcal{F}^*_2}(t, t')$. 
\end{lemma}
\begin{proof}
Assume to the contrary that $OPT$ uses the same set of filters throughout the interval $[t, t']$ and outputs $\mathcal{F}^*_1$, but $\textsc{Min}_{\mathcal{F}^*_1}(t,t') < (1-\varepsilon) \textsc{Max}_{\mathcal{F}^*_2}(t,t')$ holds. 
Then there are two nodes, $i \in \mathcal{F}^*_1$ and $j \notin \mathcal{F}^*_1$, and two times $t_1, t_2 \in [t, t']$, such that $v_i^{t_1} = \textsc{Min}_{\mathcal{F}^*_1}(t,t')$ and $v_j^{t_2} = \textsc{Max}_{\mathcal{F}^*_2}(t,t')$.
Due to the definition of a set of filters and the fact that $OPT$ has not communicated during $[t, t']$, $OPT$ must have set the filter for node $i$ to $[s_1, \infty]$, $s_1 \leq v_i^{t_1}$, and for node $j$ to $[-\infty, s_2]$, $s_2 \geq v_j^{t_2}$.
This is a contradiction to the definition of a set of filters and Observation~\ref{le:apxFilterDef}.
\end{proof}
 
At last a result from \cite{mmm} is restated in order to calculate the (exact) top-$k$ set for one time step.
\begin{lemma}\cite{mmm}
\label{le:topk}
There is an algorithm that computes the node holding the largest value using $\bigO(\log n)$ messages on expectation.
\end{lemma}

\section{Auxiliary Problem: Existence}
\label{sec:existence}
In our competitive algorithms designed and analyzed in the following, we will frequently make use of a protocol for a subproblem which we call \textsc{Existence}: 
Assume all nodes observe only binary values, i.e.\ $\forall i \in \{1, \ldots, n\} : v_i \in \{0, 1\}$.
The server is asked to decide the \textit{logical disjunction} for one fixed time step $t$. 

It is known that for $n$ nodes each holding a bit vector of length $m$ the communication complexity to decide the bit-wise disjunction is $\Omega(nm)$ in the server model \cite{phillips}. 
Observe that in our model $1$ message is sufficient to decide the problem assuming the nodes have a unique identifier between $1$ and $n$ and the protocol uses $n$ rounds. 

We prove that it is sufficient to use a constant amount of messages on expectation and logarithmic number of rounds. 
Note that the algorithm in the following lemma is a Las Vegas algorithm, i.e.\ the algorithm is always correct and the number of messages needed is based on a random process.

\begin{lemma} 
	\label{le:twoBuckets}
	There is an algorithm \textsc{ExistenceProtocol} that uses $\bigO(1)$ messages on expectation to solve the problem \textsc{Existence}.
\end{lemma}
\begin{proof}
Initially all nodes are active.
All nodes $i$ deactivate themselves, if $v_i = 0$ holds, that is, these nodes do not take part in the following process.
In each round $r = 0, 1, \ldots, \log n$  the active nodes send messages independently at random with probability $p_r \coloneqq 2^r / n$.
Consequently, if the last round $\gamma = \log n$ is reached, all active nodes $i$ with $v_i = 1$ send a message with probability 1.
As soon as at least one message was sent or the $\gamma$-th round ends, the protocol is terminated and the server can decide \textsc{Existence}.

Next, we analyze the above protocol and show that the bound on the expected number of messages is fulfilled. 
Let $X$ be the random variable for the number of messages used by the protocol and $b$ be the number of nodes $i$ with $v_i = 1$.
Note that the expected number of messages sent in round $r$ is $b\! \cdot \! p_r$ and the probability that no node has sent a message before is $\prod_{k = 0}^{r-1} \left( 1- p_k \right)^b$.

Observing that the function $f(r) = b\cdot p_r \left( 1- p_{r-1} \right)^b$ has only one extreme point and $0\leq f(r) < 2$ for $r \in [0,\log n]$, it is easy to verify that the series can be upper bounded by simple integration: 

\ifConferenceVersion
\else
\begin{align*} 
	\mathbb{E}[X] &\leq \frac{b}{n} + \sum_{r = 1}^{\log(n)} \frac{b 2^r}{n} \prod_{k = 0}^{r-1} \left( 1- \frac{2^{k}}{n} \right)^b \\
	&\leq 1 + \sum_{r = 1}^{\log(n)} \frac{b 2^r}{n} \left( 1- \frac{2^{r-1}}{n} \right)^b \\
	&\leq 1 + \int_{0}^{\log(n)} \frac{b 2^r}{n} \left( 1- \frac{2^{r-1}}{n} \right)^b dr +2
\end{align*}
\begin{align*}
	\leq& 3 + \left[ \frac{b}{(b+1) n \ln(2)} (2^r-2n) \left( 1- \frac{2^{r-1}}{n} \right)^b \right]_0^{\log n} \\
	\leq& 3 +  \frac{1}{n \ln(2)} \cdot \\ 
	&\left( (2^{\log n} - 2n) \left( 1- \frac{2^{\log n - 1}}{n} \right)^b + 2n\left( 1- \frac{2^{0 - 1}}{n} \right)^b \right)\\
	\leq& 3 +  \frac{1}{n \ln(2)} \left[ (n - 2n) \left( 1- \frac{1}{2} \right)^b + 2n \left( 1- \frac{1}{2 n} \right)^b \right]\\
	\leq& 3 +  \frac{1}{n \ln(2)} \left[ (- n)  \frac{1}{2^b}  + 2n \left( 1- \frac{1}{2 n} \right)^b \right]\\
	\leq& 3 + \frac{1}{\ln (2)} \left( 2 \left( 1- \frac{1}{2 n} \right)^b - \frac{1}{2^b} \right) \\
	\leq& 3 + \frac{1}{\ln (2)} \left( 2 - \frac{1}{2^b} \right) \leq 3 + \frac{2}{\ln(2)} \leq 6 \enspace . 
\end{align*}
\fi
\end{proof}
	
This protocol can be used for a variety of subtasks, e.g.\ validating that all nodes are within their filters, identifying that there is some filter-violation or whether there are nodes that have a higher value than a certain threshold. 
\begin{corollary}
	\label{co:validate}
	Given a time $t$.
	There is an algorithm which decides whether there are nodes which observed a filter-violation using $\bigO(1)$ messages on expectation.
\end{corollary}

\begin{proof}

For the distributed nodes to report filter-violations we use an approach based on the \textsc{ExistenceProtocol} to reduce the number of messages sent in case several nodes observe filter-violations at the same time.
The nodes apply the \textsc{ExistenceProtocol} as follows: 
Each node that is still within its filter applies the protocol using a $0$ as its value and each node that observes a filter-violation uses a $1$. 
Note that by this approach the server definitely gets informed if there is some filter-violation and otherwise no communication takes place. 
\end{proof}

The \textsc{ExistenceProtocol} can be used in combination with the relaxed definition of filters to strengthen the result for Top-$k$-Position Monitoring from $\bigO(k \log n + \log \Delta \log n)$ to $\bigO(k \log n + \log \Delta)$.
We first introduce a generic framework and then show how to achieve this bound.

\paragraph*{A generic approach} 
Throughout the paper, several of our algorithms feature similar structural properties in the sense that they can be defined within a common framework. 
Hence, we now define a generic approach to describe the calculation and communication of filters, which we then refine later.
The general idea is to only use two different filters that are basically defined by one value separating nodes in $\mathcal{F}(t)$ from the remaining nodes.
Whenever a filter-violation is reported, this value is recalculated and used to set filters properly. 
 
The approach proceeds in rounds. 
In the first round we define an initial interval $L_0$. 
In the $r$-th round, based on interval $L_r$, we compute a value $m$ that is broadcasted and is used to set the filters to $[0, m]$ and $[m, \infty]$. 
As soon as node $i$ reports a filter-violation observing the value $v_i$, the coordinator redefines the interval $L_{r+1} \coloneqq L_r \cap [0, v_i]$ if the violation is from above and $L_{r+1} \coloneqq L_r \cap [v_i, \infty]$ otherwise.
The approach finishes as soon as some (predefined) condition is satisfied.

\begin{corollary}
	\label{th:Midpoint}
	There is an algorithm that is $\bigO(k\log n + \log \Delta)$-competitive for (exact) Top-$k$-Position Monitoring.
\end{corollary}
\ifConferenceVersion
\else
\begin{proof} 
  Our algorithm proceeds in phases that are designed such that we can show that an optimal algorithm needs to communicate at least once during a phase and additionally, we can upper bound the number of messages sent by the online algorithm according to the bound on the competitiveness. 
  
  We apply the generic approach with parameters described as follows.
  The initial interval is defined as $L_0 \coloneqq [\ell, u]$, where $\ell = v_{\pi(k+1, t)}^t$, $u = v_{\pi(k, t)}^t$.
  This can be done by determining the values of the nodes holding the $k+1$ largest values using $\bigO(k \log n)$ messages on expectation.
  In the $r$-th round, based on interval $L_r$, we compute the midpoint of $L_r$ as the value $m$ which is broadcasted and used to set the filters. 
  As soon as a filter-violation is reported, the generic framework is applied.
  In case $L_r$ is empty the phase ends.
 
  Note that the distance between $u$ and $\ell$ gets halved every time a node violates its filter leading to $\bigO(\log (u_0 - \ell_0)) = \bigO(\log \Delta)$ messages on expectation per phase.
  Also, it is not hard to see that during a phase $OPT$ has communicated at least once and hence, we obtain the claimed bound on the competitiveness.
\end{proof}
\fi

\section{Competing against an Exact Adversary}
\label{sec:scatteredData}
In this section, we propose an algorithm based on the strategy to choose the nodes holding the k largest values as an output and use this set as long as it is feasible.
It will turn out that this algorithm is suitable in two scenarios: First, it performs well against an adversary who solves the Top-$k$-Position Monitoring problem (cf.\ Theorem~\ref{th:scatteredCompetitive}); second, we can use it in situations in which an adversary who is allowed to introduce some error and cannot exploit this error because the observed data leads to a unique output (cf.\ Sect.~\ref{sec:denseData}).

In particular, we develop an algorithm started at $t$ that computes the output set $\mathcal{F}_1 \coloneqq \mathcal{F}(t)$ using the protocol from Lemma~\ref{le:topk} and for all consecutive times witnesses whether $\mathcal{F}_1$ is correct or not.
Recall that while computing the set $\mathcal{F}(t)$ from scratch (cf.\ Lemma~\ref{le:topk}) is expensive in terms of communication, witnessing its correctness in consecutive rounds is cheap since it suffices to observe filter-violations (cf.\ Definition~\ref{def:filters} and Corollary~\ref{co:validate}).

The algorithm tries to find a value $m$ which partitions $\mathcal{F}_1$ from $\mathcal{F}_2$ according to the generic framework, such that for all nodes $i \in \mathcal{F}_1$ it holds $v_i \geq m$ and for all nodes $i \in \mathcal{F}_2$ it holds $v_i \leq m$.
We call such a value $m$ \emph{certificate}.

\paragraph*{Guessing OPT's Filters}
In the following we consider a time period $[t, t'']$ during which the output $\mathcal{F}(t)$  need not change.
Consider a time $t' \in [t, t'']$.
The online strategy to choose a certificate at this time contingents on the size of some interval 
\begin{center}
$L^*$ from which an offline algorithm must have chosen \\ the lower bound $\ell^*$ of the upper filter at time $t$
\end{center} such that the filters are valid throughout $[t, t']$.
The algorithm \Scattered keeps track of (an approximation of) $L^*$ at time $t'$ denoted by $L = [\ell, u]$ for which $L^* \subseteq L$ holds.
The online algorithm tries to improve the guess where OPT must have set filters by gradually reducing the size of interval $L$ (while maintaining the invariant $L^* \subseteq L$) at times it observes filter-violations.

Initially $u$ and $\ell$ are defined as follows: $u \coloneqq v_{\pi(k, t)}^t = \MIN_{\mathcal{F}_1}(t, t)$ and $\ell \coloneqq v_{\pi(k+1, t)}^t = \MAX_{\mathcal{F}_2}(t, t)$ and are redefined over time.
Although defining the certificate as the midpoint of $L = [\ell, u]$ intuitively seems to be the best way to choose $m$, the algorithm is based on four consecutive phases, each defining a different strategy.

In detail, the first phase is executed as long as the property 
\begin{equation}
\tag{P1}
\log \log u > \log \log \ell + 1
\end{equation}
holds. 
In this phase, $m$ is defined as $\ell + 2^{2^r}$ after $r$ filter-violations observed. 
If the property 
\begin{equation}
\tag{P2}
\log \log u \leq \log \log \ell +1 \wedge u > 4 \ell
\end{equation}
holds, the value $m$ is chosen to be $2^{mid}$ where $mid$ is the midpoint of $[\log \ell, \log u]$. 
Observe that $2^{mid} \in L = [\ell, u]$ holds.

The third phase is executed if property 
\begin{equation}
\tag{P3}
u \leq 4 \ \ell \wedge u > \frac{1}{1-\varepsilon}\ \ell
\end{equation}
holds and employs the intuitive approach of choosing $m$ as the midpoint of $L$.
The last phase contains the remaining case of 
\begin{equation}
\tag{P4}
u \leq \frac{1}{1-\varepsilon}\ell
\end{equation}
and is simply executed until the next filter-violation is observed using the filters $F_1 = [\ell, \infty]$ and $F_2 = [0, u]$.

In the following we propose three algorithms $\mathcal{A}_1, \mathcal{A}_2,$ and $\mathcal{A}_3$ which are executed if the respective property hold and analyze  the correctness and the amount of messages needed. 

\begin{lemma}
	\label{le:Phase1}
	Given time $t$, an output $\mathcal{F}(t)$, and an interval $L = [\ell, u]$ for which (P1) holds, there is an algorithm $\mathcal{A}_1$ that witnesses the correctness of $\mathcal{F}(t)$ until a time $t'$ at which it outputs $L' = [\ell', u']$ for which (P1) does not hold.
	The algorithm uses $\bigO(\log \log \Delta)$ messages on expectation. 
\end{lemma}
\begin{proof}
	The algorithm $\mathcal{A}_1$ applies the generic framework and
	defines the value $m$, the server broadcasts, as $m \coloneqq \ell_0 + 2^{2^r}$, where $\ell_0$ is the initial value of $\ell$.
	If $\log \log u' - \log \log \ell' \leq 1$ holds, the algorithm terminates and outputs $L'  = [\ell', u']$ with $\ell'$ and $u'$ defined as the redefinition of $\ell$ and $u$ respectively.
	
	To analyze the amount of messages needed and express it in terms of $\Delta$, observe that in the worst case the server only observes filter-violations from nodes $i \in \mathcal{F}_2$. 
	In case there is a filter-violation from above, i.e.\ a node $i \in \mathcal{F}_1$ reports a filter-violation, the condition $\log \log u' - \log \log \ell' \leq 1$ holds.
	At least in round $r = \log \log (u - \ell)$, which is by definition upper bounded by $\log \log \Delta$, the algorithm terminates.
	
	If $\mathcal{F}(t)$ is not valid at time $t'$, there are nodes $i_1 \in \mathcal{F}_1$, $i_2 \in \mathcal{F}_2$ and time points $t_1, t_2$ ($t_1 = t' \vee t_2 = t'$) for which $v_{i_1}^{t_1} < v_{i_2}^{t_2}$ holds. 
	Thus, $\mathcal{A}_1$ observed a filter-violation by either $i_1$ or $i_2$ followed by a sequence alternating between filter-violations and filter-updates.
	At some point (but still at time $t'$) $\log \log u' - \log \log \ell' \leq 1$ holds and the algorithm outputs $(\ell', u')$, proving $\mathcal{A}_1$'s correctness for time $t'$.		
\end{proof}

\begin{lemma} 
\label{le:Phase2}
For a given $\mathcal{F}(t)$ and a given interval $L = [\ell, u]$ for which (P2) holds, there is an algorithm $\mathcal{A}_2$ that witnesses the correctness of $\mathcal{F}(t)$ until a time $t'$ at which it outputs $L' = [\ell', u']$ for which (P2) does not hold.
The algorithm uses $\bigO(1)$ messages on expectation. 	
\end{lemma}
\begin{proof}
We apply the generic approach and choose the value $m$ to be broadcasted by $2^{mid}$, where $mid$ is the midpoint of $[\log \ell, \log u]$. 

To analyze the amount of messages needed, bound $L = [\ell, u]$ in terms of values that are double exponential in $2$.
To this end, let $a\in \N$ be the largest number such that $\ell \geq 2^{2^{a}}$ holds.
Now observe since (P2) holds, $u \leq 2^{2^{a+2}}$ follows.
Since the algorithm chooses the midpoint of the interval $[\log \ell, \log u]$ in order to get $m$ and halves this interval after every filter-violation, one can upper bound the number of rounds by analyzing how often the interval $[\log \ell, \log u]$ gets halved. 
This is  $[\log \ell, \log u] \subseteq \left[ \log \left( 2^{2^{a}}\right), \log \left( 2^{2^{a+3}}  \right)  \right]  = [2^{a}, 8 * 2^{a}]$ can be halved at most a constant number of times, until it contains only one value, which implies that $4 \cdot \ell > u$ holds.
\end{proof}

\begin{lemma} 
\label{le:Phase3}
For a given $\mathcal{F}(t)$ and a given interval $L = [\ell, u]$ for which (P3) holds, there is an algorithm $\mathcal{A}_3$ that witnesses the correctness of $\mathcal{F}(t)$ until a time $t'$ at which it outputs $L' = [\ell', u']$ for which (P3) does not hold.
The algorithm uses $\bigO(\log \nicefrac{1}{\varepsilon})$ messages on expectation. 	
\end{lemma}

\ifConferenceVersion
\else
\begin{proof}
The algorithm applies the generic framework and uses the midpoint strategy starting with the interval $L_0 \coloneqq [\ell, u]$.
Observe that it takes at most $\bigO\left( \log \frac{1}{\varepsilon}\right)$ redefinitions of $L$ to have the final size, no matter whether the algorithm observes only filter-violations from nodes $i \in \mathcal{F}(t)$ or $i \notin \mathcal{F}(t)$.
This together with the use of the \textsc{ExistenceProtocol} for handling filter-violations yields the needed number of messages on expectation.
The correctness follows similarly as shown for Lemma~\ref{le:Phase1}.
\end{proof}
\fi

Now we propose an algorithm started at a time $t$ which computes the output $\mathcal{F}(t)$ and witnesses its correctness until some (not predefined) time $t'$ at which the \Scattered terminates using a combination of the algorithms stated above.
Precisely the \Scattered is defined as follows:

\paragraph*{Algorithm} \Scattered
\begin{enumerate}
	\item[1.] Compute the nodes holding the $(k+1)$ largest values and define
	$\ell \coloneqq v_{k+1}^{t}$, 	$u \coloneqq v_k^{t}$ and $\mathcal{F}(t)$.	
	
	\item[2.] If (P1) holds, call $\mathcal{A}_1$ with the arguments $\mathcal{F}(t)$ and $L = [\ell, u]$.
	At the time $t'$ at which $\mathcal{A}_1$ outputs $L' = [\ell', u']$ set  $\ell \coloneqq \ell'$ and $u \coloneqq u'$. 

	\item[3.] If (P2) holds, call $\mathcal{A}_2$ with the arguments $\mathcal{F}(t)$ and $L = [\ell, u]$. 
	At the time $t'$ at which $\mathcal{A}_2$ outputs $L' = [\ell', u']$ set  $\ell \coloneqq \ell'$ and $u \coloneqq u'$. 
	
	\item[4.] If (P3) holds, call $\mathcal{A}_3$ with the arguments $\mathcal{F}(t)$ and $L = [\ell, u]$. 
	At the time $t'$ at which $\mathcal{A}_3$ outputs $L' = [\ell', u']$ set  $\ell \coloneqq \ell'$ and $u \coloneqq u'$. 
	
	\item[5.] If $u \geq \ell$ and $u \leq \frac{1}{(1-\varepsilon)} \ell$ holds, set the filters to $F_1 \coloneqq [\ell, \infty]$, $F_2 \coloneqq [0, u]$.
	At the time $t'$ at which node $i \in \mathcal{F}_2$ reports a filter-violation from below define $\ell \coloneqq v_i^{t'}$. 
	In case node $i \in \mathcal{F}_1$ reports a filter-violation from above, define $u \coloneqq v_i^{t'}$. 
	
	\item[6.] Terminate and output $(\ell, u)$.
\end{enumerate}

\begin{lemma}
\label{le:combined}
Consider a time $t$.
The algorithm \Scattered computes the top-$k$ set and witnesses its correctness until a time $t'$ at which it outputs $L = [\ell, u]$, where $\ell \leq \MAX_{\mathcal{F}_2}(t, t')$, $\MIN_{\mathcal{F}_1}(t, t') \leq u$, and $\ell > u$ holds (i.e.\ $L$ is empty).
The algorithm uses $\bigO (k \log n +  \log \log \Delta + \log \frac{1}{\varepsilon})$ messages on expectation. 
\end{lemma}

\ifConferenceVersion
\else
\begin{proof}
We first argue on the correctness of \Scattered and afterwards shortly analyze the number of messages used.

The algorithm computes in step 1. a correct output $\mathcal{F}_1$ at time $t$ by using the algorithm from Lemma~\ref{le:topk} for $k$ times. 
In consecutive time steps $t' > t$ the correctness of \Scattered follows from the correctness of algorithms $\mathcal{A}_1, \mathcal{A}_2,$ and $\mathcal{A}_3$ in steps 2. - 4.
For the correctness of step 5. observe that by setting the filters to $F_1 = [\ell, \infty]$ and $F_2 = [0, u]$ and the fact that $u \leq \frac{1}{1-\varepsilon}\ell$ holds the filters are valid.
Thus, as long as all nodes observe values which are inside their respective filters the output need not change.

At the time step $t'$ the protocol terminates and outputs $L = [\ell, u]$ it holds $u < \ell$.
Thus, there are nodes $i_1 \in \mathcal{F}_1$ and $i_2 \in \mathcal{F}_2$ and time steps $t_1, t_2 \in [t, t']$ with: $v_{i_1}^{t_1} \leq u$ and $v_{i_2}^{t_2} \geq \ell$, and thus, $v_{i_1}^{t_1} < v_{i_2}^{t_2}$. 

To argue on the number of messages observe that the first step can be executed using $\bigO(k \log n)$ number of messages. 
At the time the condition of steps 2. - 5. are checked these steps can be performed using $\bigO(k \log n)$ number of messages, by computing the nodes holding the $k+1$ largest values. 
The algorithms $\mathcal{A}_1, \mathcal{A}_2,$ and $\mathcal{A}_3$ are called at most once each thus the conditions are also checked at most once.
After executing step 5. the algorithm terminates which leads to the result on the number of messages as stated above.
\end{proof}
\fi

\begin{theorem}
\label{th:scatteredCompetitive}
The algorithm \Scattered has a competitiveness of $\bigO(k \log n + \log \log \Delta + \log \frac{1}{\varepsilon})$ allowing an error of $\varepsilon$ compared to an optimal offline algorithm that solves the exact Top-$k$-Position Monitoring problem. 
\end{theorem}
\begin{proof}
The correctness of \Scattered and the number of messages follow from Lemma~\ref{le:combined}.
Now we argue that OPT had to communicate at least once in the interval $[t, t']$ during which \Scattered was applied.
If OPT communicated, the bound on the competitiveness directly follows.
Now assume that OPT did not communicate in the interval $[t,t']$.
We claim that the interval $L$ maintained during \Scattered always satisfies the invariant $L^* \subseteq L$.
If this claim is true, we directly obtain a contradiction to the fact that OPT did not communicate because of the following reasons.
On the one hand, because OPT has to monitor the exact Top-$k$-Positions, OPT chooses the same set of nodes $\mathcal{F}^* = \mathcal{F}_1$ which was chosen by the online algorithm.
On the other hand, at the time $t'$ the algorithm \Scattered terminates, $u' < \ell'$ holds. 
Thus, the interval $L'$ is empty and since $L^* \subseteq L'$ holds, it follows that $L^*$ is empty and hence, OPT must have communicated.

We now prove the claim.
Recall that \Scattered is started with an interval $L$ that fulfills $L^* \subseteq L$ by definition.
To show that $L^* \subseteq L$ holds during the entire interval $[t,t']$, it suffices to argue that each of the previous algorithms makes sure that when started with an interval $L$ such that $L^* \subseteq L$, it outputs $L'$ with $L^* \subseteq L'$.
Our following reasoning is generic and can be applied to the previous algorithms.
Consider the cases in which filter-violations are observed and hence the interval $L$ is modified:
If a filter-violation from below happened at a time $t_1 > t$, there is a node $i \in \mathcal{F}_2$ with a value $v_i^{t_1} > \ell'$ and thus, $\ell^* > \ell'$ holds.
If a filter-violation from above happened at a time $t'$, there is a node $i \in \mathcal{F}_1$ with a value $v_i^{t'} < u'$ and thus, $u^* < u'$ holds.
This case-distinction leads to the result, that $L^*$ has to be a subset of $[\ell', u']$.
\end{proof}

\section{Competing against an Approximate Adversary}
\label{sec:denseData}
In this section, we study the case in which the adversary is allowed to use an approximate filter-based offline algorithm, i.e.\ one that solves \etopk.
Not surprisingly, it turns out that it is much more challenging for online than for offline algorithms to cope with or exploit the allowed error in the output. 
This fact is formalized in the lower bound in Theorem~\ref{le:lowerBound}, which is larger than previous upper bounds for the exact problem.
However, we also propose two online algorithms that are competitive against offline algorithms that are allowed to have the same error $\varepsilon$ and a smaller error $\varepsilon' \leq  \frac{\varepsilon}{2}$, respectively.

\subsection{Lower Bound for Competitive Algorithms}
We show a lower bound on the competitiveness proving any online algorithm has to communicate at least $(\sigma - k)$ times in contrast to an offline algorithm which only uses $k+1$ messages.
Recall that the adversary generates the data streams and can see the filters communicated by the server.
Note that as long as the online and the offline algorithm are allowed to make use of an error $\varepsilon \in (0, 1)$ the lower bound holds, even if the errors are different.

\begin{theorem}
\label{le:lowerBound}
Any filter-based online algorithm which solves the \etopk problem and is allowed to make use of an error of $\varepsilon \in (0, 1)$ has a competitiveness of $\Omega \left(\nicefrac{\sigma}{k} \right)$ compared to an optimal offline algorithm which is allowed to use a (potentially different) error of $\varepsilon' \in (0, 1)$.
\end{theorem}

\begin{proof}
	Consider an instance in which the observed values of $\sigma \in [k+1, n]$ nodes are equal to some value $y_0$ (the remaining $n-\sigma$ nodes observe smaller values) at time $t = 0$ and the following adversary:
	In time step $r=0,1,\ldots, n-k$, the adversary decides to change the value of one node $i$ with $v_i^r=y_0$ to be $v_i^{r+1}=y_1<(1-\varepsilon) \cdot y_0$ such that a filter-violation occurs.
	Observe that such a value $y_1$ exists if $\varepsilon < 1$ holds and a node $i$ always exists since otherwise the filters assigned by the online algorithm cannot be feasible. 
	Hence, the number of messages sent by the online algorithm until time step $n-k$ is at least $n-k$.
	In contrast, the offline algorithm knows the $n-k$ nodes whose values change over time and hence, can set the filters such that no filter-violation happens.
	The offline algorithm sets two different filters: 
	One filter $F_1 = [y_0, \infty]$ for those $k$ nodes which have a value of $y_0$ at time step $n-k$ using $k$ messages and one filter $F_2 = [0, y_0]$ for the remaining $n-k$ nodes using one broadcast message.
	By essentially repeating these ideas, the input stream can be extended to an arbitrary length, obtaining the lower bound as stated. 
\end{proof}

\subsection{Upper Bounds for Competitive Algorithms}
\label{sec:upperBounds}
Now we propose an algorithm \Dense and analyze the competitiveness against an optimal offline algorithm in the setting that both algorithms are allowed to use an error of $\varepsilon$.

The algorithm \Dense is started a time $t$. 
For sake of simplicity we assume that the $k$-th and the $(k+1)$-st node observe the same value $z$, that is $z \coloneqq v_{\pi(k, t)}^t = v_{\pi(k+1, t)}^t$. 
However, if this does not hold we can define the filters to be $F_1 = [v_{\pi(k+1, t)}^t, \infty]$ and $F_2 = [0, v_{\pi(k, t)}^t]$ until a filter-violation is observed at some time $t'$ using $\bigO(k \log n)$ messages on expectation. 
If the filter-violation occurred from below define $z \coloneqq v_{\pi(k, t)}^t$ and if a filter-violation from above is observed define $z \coloneqq v_{\pi(k+1, t)}^t$.

The high-level idea of \Dense is similar to the \Scattered to compute a guess $L$ on the lower endpoint of the filter of the output $\mathcal{F}^*$ of OPT (assuming OPT did not communicate during $[t, t']$) for which the invariant $\ell^* \in L^* \subseteq L_r$ holds.
The goal of \Dense is to halve the interval $L$ while maintaining $\ell^* \in L$ until $L = \emptyset$ and thus show that no value exists which could be used by OPT. 

To this end, the algorithm partitions the nodes into three sets. 
Intuitively speaking, the first set which we call $V_1$ contains those nodes which have to be part of the optimal output, $V_3$ those nodes that cannot be part of any optimal output and $V_2$ the remaining nodes.
The sets change over time as follows.
Initially $V_1^t$ contains those nodes that observes a value $v_i^t > \frac{1}{1-\varepsilon}z$.
Since the algorithm may discover at a time $t' > t$ that some node $i$ has to be moved to $V_1^{t'+1}$ which also contains all nodes from previous rounds, i.e.\ $V_1^{t'} \subseteq V_1^{t'+1}$.
On the other hand $V_3^t$ initially contains the nodes which observed a value $v_i^t < (1-\varepsilon)z$.
Here also the algorithm may discover at a time $t' > t$ that some node $i$ has to be moved to $V_3^{t'+1}$ which (similar to $V_1$) contains nodes from previous rounds.
At the time $t$ the set $V_2^t$ simply contains the remaining nodes $\{1, \ldots, n\} \setminus (V_1^t \cup V_3^t)$ and its cardinality will only decrease over time.

In the following we make use of sets $S_1$ and $S_2$ to indicate that nodes in $V_2$ may be moved to $V_1$ or $V_3$ depending on the values observed by the remaining nodes in $V_2$.
Nodes in $S_1$ observed a value larger than $z$ but still not that large to decide to move it to $V_1$ and similarly nodes in $S_2$ observed smaller values than $z$ but not that small to move it to $V_3$. 

Next we propose the algorithm \Dense in which we make use of an algorithm \Subprotocol for the scenario in which some node $i$ exists that is in $S_1$ and in $S_2$. 
At a time at which the \Subprotocol terminates it outputs that $\ell^*$ has to be in the lower half of $L$ or in the upper half of $L$ thus, the interval $L$ gets halved (which initiates the next round) or moves one node from $V_2$ to $V_1$ or $V_3$.
Intuitively speaking \Subprotocol is designed such that, if OPT did not communicate during $[t, t']$, where $t$ is the time the \Dense is started and $t'$ is the current time step, the movement of one node $i \in V_2$ to $V_1$ or $V_3$ implies that $i$ has necessarily to be part of $\mathcal{F}^*$ or not.
For now we assume the algorithm \Subprotocol to work correctly as a black box using $SUB(n, |L|)$ number of messages.

Note that in case $L_r$ contains one value and gets halved, the interval $L_{r+1}$ is defined to be empty.
In case the algorithm observes multiple nodes reporting a filter-violation the server processes one violation at a time in an arbitrary order. 
Since the server may define new filters after processing a violation one of the multiple filter-violations may be not relevant any longer, thus the server simply ignores it.
\vspace{-0.3cm}
\paragraph*{Algorithm: \Dense}~\newline 
\vspace{-0.5cm}
\begin{itemize}
\item[1.]
	Define $z \coloneqq v_k^t = v_{k+1}^t$ and the following sets:\\
	$V_1 \coloneqq \{i \in \{1, \ldots, n\} \mid v_i^t > \frac{1}{1-\varepsilon}z \}, \\
	V_3 \coloneqq \left\{i \in \{1, \ldots, n\} \mid v_i^t < (1-\varepsilon)z\right\},  \\
	V_2 \coloneqq \{1, \ldots, n\} \setminus (V_1 \cup V_3).$

	Define an interval $L_0 \coloneqq [(1-\varepsilon)z, z]$ and 
	define sets $S_1, S_2$ of nodes which are initially empty and use $S$ to denote $S_1 \cup S_2$.
	Set $r \coloneqq 0$ indicating the round.

\item[2.] The following \textbf{rules} are applied for (some) round $r$: \\
	Let $\ell_r$ be the midpoint of $L_r$ and $u_r \coloneqq \frac{1}{1-\varepsilon} \ell_r$\newline
	For a node $i$ the filter is defined as follows: \\
	If $i \in V_1$, $F_i \coloneqq [\ell_r, \infty]$;\\ 
	If $i \in V_2 \cap S_1$, $F_i \coloneqq [\ell_r, \frac{1}{1-\varepsilon}z]$. \\
	if $i \in V_2 \setminus S$, $F_i \coloneqq [\ell_r, u_r]$; \\	
	If $i \in V_2 \cap S_2$, $F_i \coloneqq [(1-\varepsilon) z, u_r]$. \\
	if $i \in V_3$, $F_i \coloneqq [0, u_r]$. \\
	The output $\mathcal{F}(t)$ is defined as $V_1 \cup (S_1 \setminus S_2)$ and $k - |V_1 \cup (S_1 \setminus S_2)|$ many nodes from $V_2 \setminus S_2$.

\item[3.] Wait until time $t'$, at which some node $i$ reports a filter-violation:
	\begin{enumerate}
	\item[a.] \textbf{If} $i \in V_1$, \textbf{then} set $L_{r+1}$ to be the lower half of $L_r$ and define $S_2 \coloneqq \emptyset$.
	
	\item[b.] \textbf{If} $i \in (V_2 \setminus S)$ violates its filter from below \textbf{then}
		\begin{enumerate}
			\item[b.1.] \textbf{If} the server observed strictly more than $k$ nodes with larger values than $u_r$ 
			\textbf{then} set $L_{r+1}$ to be the upper half of $L_r$ and define $S_1 \coloneqq \emptyset$.
			\item[b.2.] \textbf{else} add $i$ to $S_1$ and update $i$'s filter.
		\end{enumerate}

	\item[c.] \textbf{If} $i \in S_1 \setminus S_2$ violates its filter \textbf{then}
	\begin{enumerate}
		\item[c.1.] \textbf{If} $i$ violates its filter from below 
		\textbf{then} move $i$ from $S_1$ and $V_2$ to $V_1$ and update $i$'s filter.
		\item[c.2.] \textbf{else} add $i$ to $S_2$ and call \Subprotocol.
	\end{enumerate}

	\item[d.] \textbf{If} the server observed $k$ nodes with values $v_i > u_r$ and $n-k$ nodes with values $v_i < \ell_r$ \textbf{then} call \Scattered
	
	\item[e.] \textbf{If} $L_{r+1}$ was set \textbf{if} is empty, end the protocol, otherwise increment $r$, update $u_r$, $\ell_r$, all filters using the rules in 2., and goto step 3.   

	~\\
	--- And their symmetric cases ---
	\item[a'.] \textbf{If} $i \in V_3$ \textbf{then} set $L_{r+1}$ to be the upper half of $L_r$ and define $S_1 \coloneqq \emptyset$.
	
	\item[b'.] \textbf{If} $i \in (V_2 \setminus S)$ violates its filter from above \textbf{then}
	\begin{enumerate}
		\item[b'.1.] \textbf{If} the server observed strictly more than $n - k$ nodes with smaller values than $\ell_r$ \textbf{then} set $L_{r+1}$ to the lower half of $L_r$ and define $S_2 \coloneqq \emptyset$.
		\item[b'.2.] \textbf{else} add $i$ to $S_2$.
	\end{enumerate}
		
	\item[c'.] \textbf{If} $i \in S_2 \setminus S_1$ violates its filter \textbf{then}
			\begin{enumerate}
				\item[c'.1.] \textbf{If} $i$ violates its filter from above \\ 
				\textbf{then} delete $i$ from $S_2$, delete $i$ from $V_2$, and add $i$ to $V_3$.
				\item[c'.2.] \textbf{else} add $i$ to $S_1$ and call \Subprotocol.
			\end{enumerate} 
	\end{enumerate}
\end{itemize}

We analyze the correctness of the protocol in the following lemma and the number of messages used in Lemma~\ref{le:denseMsg}.
We prove that OPT communicated at least once in Lemma~\ref{le:optOnce}.

\begin{lemma}
\label{le:denseCorrect}
The protocol \textsc{DenseProtocol} computes a correct output $\mathcal{F}(t')$ at any time $t'$.
\end{lemma}
\begin{proof} 
By definition the output consists of nodes from $V_1$, $S_1$ and (arbitrary) nodes from $V_2 \setminus S_2$ (cf.\ step 2.). 
Observe that by definition of the filters of the nodes in these subsets, the minimum of all lower endpoints of the filters is $\ell_r$ following the rules in step 2. 
Also observe that the maximum of all upper endpoints of the filters of the remaining nodes is $u_r$. 
Since by definition $u_r = \frac{1}{1-\varepsilon} \ell_r$ holds, the values observed by nodes $i \in \mathcal{F}_1$ are (lower) bounded by $\ell_r$ and nodes $i \in \mathcal{F}_2$ are (upper) bounded by $u_r$, thus the overlap of the filters is valid.

Now we argue that there are at least $k$ nodes in the set $V_1 \cup S_1 \cup V_2 \setminus S_2$.
To this end, assume to the contrary that $t'$ is the first time step at which strictly less than $k$ nodes are in the union of these sets. 
Now observe that the cases in the \Dense in which nodes are deleted from one of $V_1, S_1$ or $V_2 \setminus S_2$ are 3.c.1., 3.c.2., and 3.b'.2..

Observe that in step 3.c.1. the algorithm moves $i$ from $S_1$ and $V_2$ to $V_1$ and thus $i$ is again part of the output and does not change the cardinality.
In step 3.c.2. the node $i$ is added to $S_2$ and \Subprotocol is called afterwards.
At this time $t'$ node $i$ is (again) part of the output of \Subprotocol and thus there are sufficiently many nodes to choose as an output which is a contradiction to the assumption.
In the remaining case 3.b'.2. \Dense adds $i$ to $S_2$.
However, since at time $t'$ strictly less than $k$ nodes are in $V_1 \cup S_1 \cup (V_2 \setminus S_2)$, there are strictly more than $n-k$ nodes in $S_2 \cup V_3$ and thus, the algorithm would execute step 3.b'.1. instead.
This leads to a contradiction to the assumption.
By these arguments the correctness follows.
\end{proof}

\begin{lemma}
\label{le:denseMsg}
The protocol \textsc{DenseProtocol} uses at most $\bigO(k \log n + \sigma \log (\varepsilon v_k ) + (\sigma + \log(\varepsilon v_k))  \cdot SUB(\sigma, |L|))$ messages on expectation.
\end{lemma}
\begin{proof} 
Initially the algorithm computes the top-$k$ set and probes all nodes which are in the $\varepsilon$-neighborhood of the node observing the $k$-th largest value, using $\bigO(k \log n + \sigma)$ messages on expectation.

During each round $r$ each node can only violate its filter at most constant times without starting the next round $r+1$ or leading to a call of \Subprotocol based on the following simple arguments: 
All nodes $i$ in $V_1$ or $V_3$ directly start the next round $r+1$ after a filter-violation. 
Now fix a node $i \in V_2$ and observe that if it is not contained in $S_1$ and $S_2$ it is added to $S_1$ if a filter-violation from below or to $S_2$ if a filter-violation from above is observed. 
At the time this node $i$ observes a filter-violation in the same direction (i.e.\ from below if it is in $S_1$ and from above if it is in $S_2$) it is added to $V_1$ or $V_3$. 
In these cases the next filter-violation will start the next round.
The last case that remains is that it is added to both sets, $S_1$ and $S_2$.
Observe that the \Subprotocol is called and starts the next round or decides on one node (which may be different from the fixed node $i$) to be moved to $V_1$ or $V_3$.

Observe that at most $\sigma + 1$ nodes can perform filter-violations without starting the next round since each node from $V_1$ or $V_3$ directly starts the next round and the number of nodes in $V_2$ is bounded by $\sigma$.
Furthermore observe that after each round the interval $L$ is halved thus, after at most $\log |L_0| + 1$ rounds the set $L_r$ is empty. 

Now focus on the \Subprotocol which also halves $L$ after termination or decides on one node $i \in V_2^{t'}$ to be moved to $V_1^{t'+1}$ or $V_3^{t'+1}$.
Thus, it can be called at most $\sigma + \log (\varepsilon v_k)$ times, leading to the result as stated above.
\end{proof}

\paragraph*{The \Subprotocol}
\ifConferenceVersion
\else
We propose an algorithm which is dedicated for the case in the execution of \Dense that one node $i$ was added to $S_1$ and to $S_2$.

$ \in V_2 \setminus S$ reported a filter-violation from below and from above and thus gets added to $S_1$ and to $S_2$ (in an arbitrary order). 
In detail, it has observed a value which is larger than $u_r$ and a value which is smaller than $\ell_r$. 
As a short remark, if $i \in \mathcal{F}^*$ would hold, then $\ell^* \leq \ell_r$ follows and on the other hand if $i \notin \mathcal{F}^*$ holds, then $\ell^* \geq \ell_r$ follows, but in \Dense cannot decide $i \in \mathcal{F}^*$ in steps 3.c.2. or 3.c'.2.

\paragraph*{Algorithm: \Subprotocol}~\newline 
\vspace{-0.3cm}
\begin{itemize}
\item[1.] Define an interval $L'_0 \coloneqq L_r \cap [(1-\varepsilon)z, \ell_r]$, $S_1' \coloneqq S_1$, and $S_2' \coloneqq \emptyset$.
Set $r' \coloneqq 0$ indicating the round.

\item[2.]  The following \textbf{rules} are applied for (some) round $r'$: \\
Let $\ell'_r$ be the midpoint of $L'_{r'}$ and $u'_{r'} \coloneqq \frac{1}{1-\varepsilon} \ell'_{r'}$. \newline
For a node $i$ the filter is defined as follows: \\
If $i \in V_1$, $F'_i \coloneqq F_i$; \\
If $i \in V_2 \cap (S_1' \setminus S_2')$, $F_i' \coloneqq [\ell_r, \frac{1}{1-\varepsilon}z]$. \\
If $i \in V_2 \cap S_1' \cap S_2'$, $F'_i \coloneqq [\ell_{r'}', \frac{1}{1-\varepsilon} z]$; \\
if $i \in V_2 \setminus S'$, $F'_i \coloneqq [\ell_r, u'_{r'}]$.\\
if $i \in V_2 \cap (S_2' \setminus S_1')$, $F'_i \coloneqq [(1-\varepsilon)z, u'_{r'}]$; \\
if $i \in V_3$, $F'_i \coloneqq [0, u'_{r'}]$;

The output $\mathcal{F}(t)$ is defined as $V_1 \cup (S'_1 \setminus S'_2) \cup (S'_1 \cap S'_2)$ and sufficiently many nodes from $V_2 \setminus S'_2$.

\item[3.] Wait until time $t'$, at which node $i$ reports a filter-violation:
\begin{enumerate}
	\item[a.] \textbf{If} $i \in V_1$, 
	\textbf{then} terminate \Subprotocol and set $L_{r+1}$ to be the lower half of $L_r$.

	\item[b.] \textbf{If} $i \in (V_2 \setminus S')$ violates its filter from below 
	\begin{enumerate}
		\item[b.1.] \textbf{If} the server observed strictly more than $k$ nodes with larger values than $u_r$ \textbf{then}
		\begin{itemize}
			\item set $L'_{{r'}+1}$ to be the upper half of $L'_{r'}$ and redefine $S'_1 \coloneqq S_1$. 
			\item \textbf{If} $L'_{r'+1}$ is defined to the empty set \textbf{then} terminate \Subprotocol and define the last node $i$ which was in $S'_1 \cap S'_2$ and observed a filter-violation from above to be moved to $V_3$. 
			If such a node does not exist the node $i \in S_1 \cap S_2$ moves to $V_3$.			
		\end{itemize}

		\item[b.2.] \textbf{Else} add $i$ to $S'_1$.
	\end{enumerate}
	
	\item[c.] \textbf{If} $i \in S'_1 \setminus S'_2$ violates its filter
	\begin{enumerate}
		\item[c.1.] \textbf{If} $i$ violates its filter from below \textbf{then} move $i$ from $V_2$ and $S'_1$ to $V_1$.
		\item[c.2.] \textbf{Else} add $i$ to $S'_2$ and update $i$'th filter.
	\end{enumerate}
	
	\item[d.] \textbf{If} $i \in S'_1 \cap S'_2$ violates its filter
	\begin{enumerate}
		\item[d.1.] \textbf{If} $i$ violates from below \textbf{then} move $i$ to $V_1$ terminate the \Subprotocol.
		\item[d.2.] \textbf{else} 
		\begin{itemize}
			\item define $L'_{r'+1}$ to be the lower half of $L'_{r'}$ and redefine $S'_2 \coloneqq \emptyset$.
			\item \textbf{If} $L'_{r'+1}$ is defined to be the empty set \textbf{then} terminate \Subprotocol and move $i$ to $V_3$.
		\end{itemize}
	\end{enumerate} 
	
	\item[e.] \textbf{If} the server observed $k$ nodes with values $v_i > u_r$ and $n-k$ nodes with values $v_i < \ell_r$ \textbf{then} call \Scattered
		
	\item[f.] \textbf{If} $L'_{r'+1}$ was set increment $r'$, update $u'_{r'}$, $\ell'_{r'}$, all filters using the rules in 2., and goto step 3.   
	\\
		
	--- And their symmetric cases ---
	\item[a'.] \textbf{If} $i \in V_3$, \textbf{then} 
	\begin{itemize}
		\item set $L'_{r'+1}$ to be the upper half of $L'_{r'}$ and redefine $S'_1 \coloneqq S_1$.
		\item \textbf{If} $L'_{r'+1}$ is defined to the empty set \textbf{then} terminate \Subprotocol and define the last node $i$ which was in $S'_1 \cap S'_2$ and observed a filter-violation from above to be moved to $V_3$. 
		If such a node does not exist the node $i \in S_1 \cap S_2$ moves to $V_3$.
	\end{itemize}
		
	\item[b'.] \textbf{If} $i \in (V_2 \setminus S')$ violates its filter from above 
	\begin{enumerate}
		\item[b'.1.] \textbf{If} the server observed strictly more than $n - k$ nodes with a value less than $\ell_r$, \textbf{then} terminate \Subprotocol and set $L_{r+1}$ to be the lower half of $L_r$.
		\item[b'.2.] \textbf{else} add $i$ to $S'_2$.
	\end{enumerate}
	
	\item[c'.] \textbf{If} $i \in S'_2 \setminus S'_1$ 
	\begin{enumerate}
		\item[c'.1.] \textbf{If} $i$ violates its filter from above \textbf{then} move $i$ from $V_2$ and $S'_2$ to $V_3$.
		\item[c'.2.] \textbf{else} add $i$ to $S'_1$ and update $i$'th filter.
	\end{enumerate}
\end{enumerate}
\end{itemize}

\fi

\begin{lemma}
\label{le:subCorrect}
The protocol \Subprotocol computes a correct output $\mathcal{F}(t')$ at any time $t'$ at which a node $i \in S_1 \cap S_2$ exists.
\end{lemma}
\ifConferenceVersion
\else
\begin{proof}
By definition the output consists of nodes from $V_1$, $S'_1 \setminus S_2'$, $S'_1 \cap S'_2$ and (arbitrary) nodes from $V_2 \setminus S'_2$ (cf.\ step 2.). 
Observe that by definition of the filters of the nodes in these subsets, the minimum of all lower endpoints of the filters is $\ell'_{r'}$ (in case the node is in $S_1$ and in $S_2$) following the rules in step 2. 
Also observe that the maximum of all upper endpoints of the filters of the remaining nodes (in subsets $V_2 \setminus S'$, $S'_2 \setminus S'_1$ or $V_3$) is $u'_{r'}$. 
Since by definition $u'_{r'} = \frac{1}{1-\varepsilon} \ell'_{r'}$ holds, the values observed by nodes $i \in \mathcal{F}_1$ are (lower) bounded by $\ell'_{r'}$ and nodes $i \in \mathcal{F}_2$ are (upper) bounded by $u'_{r'}$ thus, the overlap of the filters is valid.

Now we argue that there are at least $k$ nodes in the sets $V_1$, $S_1 \setminus S_2$, $S_1 \cap S_2$, and $V_2 \setminus S_2$.
To this end, simply assume to the contrary that at a time $t'$ there are strictly less than $k$ nodes in the union of the sets. 
It follows that at this time $t'$, the algorithm has observed that there are strictly more than $n-k$ nodes with a value smaller than $\ell'_{r'}$.
Thus, the algorithm would continue (compare case b'.1.) with a lower value of $\ell_r$ or, in case the interval $L_r$ is empty, terminates (which is a contradiction). 

By these arguments the correctness follows.	
\end{proof}
\fi

\begin{lemma}
\label{le:subMsg}
The protocol \Subprotocol uses at most $\bigO(\sigma \log |L|)$ messages on expectation.
\end{lemma}
\ifConferenceVersion
\else
\begin{proof}
During each round $r'$ each node can only violate its filter at most constant times without starting the next round $r'+1$ based on the following simple arguments: 
All nodes $i$ in $V_1$ or $V_3$ directly start the next round $r'+1$ after a filter-violation. 
Now fix a node $i \in V_2$ and observe that if it is not contained in $S'_1$ and $S'_2$ it is added to $S'_1$ if a filter-violation from below or to $S'_2$ if a filter-violation from above is observed. 
At the time this node $i$ observes a filter-violation in the same direction (i.e.\ from below if it is in $S'_1$ and from above if it is in $S'_2$) it is added to $V_1$ or $V_3$. 
In these cases the next filter-violation will start the next round.
The last case that remains is that it is added to both sets, $S'_1$ and $S'_2$.
Observe that \Scattered terminates if $i \in S'_1 \cap S'_2$ violates its filter from below (and moves $i to V_1$).
Otherwise $i$ violates its filter from above \Subprotocol starts the next round $r'+1$.

Observe that at most $\sigma + 1$ nodes can perform filter-violations without starting the next round since each node from $V_1$ or $V_3$ directly starts the next round ($r + 1$ from the \Dense or $r'+1$ this protocol) and the number of nodes in $V_2$ is bounded by $\sigma$.

Furthermore observe that after each round the interval $L'$, the guess of OPTs lower endpoint of the upper filter, is halved. 
The range of $L'$ is upper bounded by the range of $L$ thus, after at most $\log  |L| + 1$ rounds the set $L'$ is empty.
\end{proof}
\fi

\begin{lemma}
\label{le:subprotocol}
Given a time point $t$ at which \Subprotocol is started.
At the time $t'$ which \Subprotocol terminates, there is one node $i$ that is moved from $V_2$ to $V_1$ or $V_3$ or the interval $L_r$ (from \Dense) is halved correctly.
\end{lemma}
\ifConferenceVersion
\else
\begin{proof}
Focus on the cases in which $L'$ is halved or there is a decision on a node $i$ to move to $V_1$ or $V_3$ (cf. cases 3.b.1., 3.d.1. 3.d.2., 3.a'., and 3.c'.1.).

In step 3.b.1. the server observed at the time $t'$ a filter-violation from $i \in V_2 \setminus S'$ and there are (strictly) more than $k$ nodes observed with a larger value than $u'_{r'}$. 
Observe that in this case for all subsets $\mathcal{S}$ with $k$ elements there exists one node $i \notin \mathcal{S}$ which observed a value $v_i \geq u'_{r'}$, thus no matter which set is chosen by OPT, for the upper bound $u^*$ for nodes $i \notin \mathcal{F}^*$ it holds: $u^* \geq u'_{r'}$, and since $u'_{r'} = \frac{1}{1-\varepsilon} \ell'_{r'}$ holds, it follows $\ell^* \geq \ell'_{r'}$.
Furthermore if $L'_{r'+1}$ was defined as the empty set, and a node $i \in S'_1 \cap S'_2$ exists, observe that $i$ gets a value $v_i \leq \ell'_{r'}$ and since in this case $u^* \geq u'_{r'}$ holds, $i \notin \mathcal{F}^*$ follows.
If such a node $i$ does not exist during the execution of \Subprotocol, the node $i \in S_1 \cap S_2$ which initiated the \Subprotocol can be decided to move to $V_3$ since during the execution of \Subprotocol the interval $L'$ is only halved to the upper half, thus $i \in S_1 \cap S_2$ observed a value $v_i < \ell_r = \ell'_{r'}$ and since $u^* \geq u'_{r'}$ holds, this $i \notin \mathcal{F}^*$ follows.

In step 3.d.1. the node $i$ observed a value $v_i$ which is larger than $\frac{1}{1-\varepsilon} z$ and thus has to be part of $\mathcal{F}^*$. 

In step 3.d.2. the node $i$ observed a value $v_i < \ell'_{r'}$. 
If during the execution of \Subprotocol the set $L'$ was defined as the upper half at least once then there was a node $j \in V_3$ or strictly more than $k$ nodes which observed a larger value than $u'_{r'}$. 
It follows, that this $i$ cannot be part of $\mathcal{F}^*$. 
In case during the execution of \Subprotocol the set $L'$ is alway defined to the lower half, then $\ell'_{r'}$ is the lower end of $L$ and since node $i$ observed a value strictly smaller than $\ell'_{r'}$ it cannot be part of $\mathcal{F}^*$.

The arguments for case 3.a'. are similar to 3.b.1.

For the remaining case 3.c'.1. simply observe that $i$ observed a smaller value than $(1-\varepsilon)z$ thus $i$ cannot be part of $\mathcal{F}^*$ follows.

First, focus on the steps in which $L$ is halved and observe that steps 3.a. and 3.b'.1. are the same cases as in the \Dense.
\end{proof}
\fi

\begin{lemma}
\label{le:optOnce}
Given a time point $t$ at which \Dense is started.
Let $t'$ be the time point at which \Dense terminates. 
During the time interval $[t, t']$ OPT communicated at least once.
\end{lemma}
\begin{proof}
We prove that OPT communicated by arguing that $\ell^*$, the lower endpoint of the upper filter, i.e.\ the filter for the output $\mathcal{F}^*$, is in the guess $L_r$ at each round $r$ ($\ell^* \in L^* \subseteq L_r$).
Hence we show that although if we halve the interval $L_r$, the invariant $\ell^* \in L^* \subseteq L_r$ is maintained all the time of the execution of \Dense and possible calls of \Subprotocol. 

In the following we assume to the contrary that OPT did not communicate throughout the interval $[t, t']$. 
We first argue for the execution of \Dense and assume that the invariant by calls of \Subprotocol hold by Lemma~\ref{le:subprotocol}.

First focus on the \Dense, which halves the interval $L_r$ in steps 3.a., 3.b.1., 3.a'., and 3.b'.1.:

In step 3.a. in which a node $i \in V_1$ violates its filter from above and observes a value $v_i < \ell_r$, it holds: $i \in \mathcal{F}^*$ thus, $\ell^* < \ell_r$ follows. 

In step 3.b.1. there are (strictly) more than $k$ nodes with a larger value than $u_r$.
It follows that for all subsets $\mathcal{S}$ (with $k$ elements) there is one node $i \notin \mathcal{S}$ observing a value larger than $u_r$ and thus, $\ell^* \geq (1-\varepsilon) u_r = \ell_r$ holds.  

The case 3.a'. (which is symmetric to 3.a.) is executed if a node $i \in V_3$ observed a filter-violation ($v_i > u_r$) which implies that the upper endpoint $u^*$ of filter $F_2$ is larger than $v_i$ and thus, $\ell^* \geq (1-\varepsilon) u_r = \ell_r$.

In step 3.b'.1. (which is symmetric to 3.b.1.) there are (strictly) more than $n-k$ nodes with a smaller value than $\ell_r$.
It follows that for all subsets $\mathcal{S}$ (with $k$ elements) there is one node $i \in \mathcal{S}$ observing a value smaller than $\ell_r$ and thus, $\ell^* \leq \ell_r$ holds.
\end{proof}

\begin{theorem}
There is an online algorithm for \etopk~which is $\bigO(\sigma^2 \log(\varepsilon v_k) + \sigma \log^2 (\varepsilon v_k) + \log \log \Delta + \log \frac{1}{\varepsilon})$\mbox{-competitive} against an optimal offline algorithm which may  use an error of $\varepsilon$.
\end{theorem}
\begin{proof}
The algorithm works as follows. 
At time $t$ at which the algorithm is started, the algorithm probes the nodes holding the $k+1$ largest values.
If $v_{\pi(k+1, t)}^t < (1-\varepsilon)v_{\pi(k, t)}^t$ holds, the algorithm \Scattered is called. 
Otherwise the algorithm \Dense is executed.
After termination of the respective call, the procedure starts over again.

Observe that if the condition holds, there is only one unique output and thus, the \Scattered monitors the Top-$k$-Positions satisfiying the bound on the competitiveness as stated in Theorem~\ref{th:scatteredCompetitive}. 
If the condition does not hold, there is at least one value in the $\varepsilon$-neighborhood of $v_{\pi(k, t)}^t$ and thus, the \Dense monitors the approximated Top-$k$-Positions as analyzed in this section.

The number of messages used is simply obtained by adding the number of messages used by the respective algorithms as stated above.
\end{proof}

To obtain the upper bounds stated at the beginning, we upper bound $\sigma$ by $n$ and $v_k$ by $\Delta$:
$\bigO(n^2 \log(\varepsilon \Delta) + n \log^2 (\varepsilon \Delta) + \log \log \Delta + \log \frac{1}{\varepsilon})$.
Note that for constant $\varepsilon$ we obtain a slightly simpler bound of
$\bigO(n^2 \log \Delta + n \log^2 \Delta )$ on the competitiveness.

\begin{corollary}
There is an online algorithm for \etopk~which is $\bigO(\sigma + k \log n + \log \log \Delta + \log \frac{1}{\varepsilon})$-competitive against an optimal offline algorithm which may use an error of $\varepsilon' \leq \frac{\varepsilon}{2}$.
\end{corollary}

\begin{proof}
The algorithm works as follows. 
At the initial time step $t$ the algorithm probes the nodes holding the $k+1$ largest values.
If $v_{\pi(k+1, t)}^t < (1-\varepsilon)v_{\pi(k, t)}^t$ holds the algorithm \Scattered is called. 

Otherwise the online algorithm simulates the first round of the \Dense, that is nodes are partitioned into $V_1, V_2,$ and $V_3$ and the filters are defined as proposed (cf step 2. of \Dense).
Here all nodes with values larger than $\frac{1}{1-\varepsilon}(1-\frac{\varepsilon}{2}) z$ are directly added to $V_1$ instead of adding to $S_1$, and nodes observing values smaller than $(1-\frac{\varepsilon}{2}) z$ are added to $V_3$.
Furthermore, if a filter-violation from some node $i \in V_2$ is observed, it is directly moved (deleted from $V_2$ and added) to $V_1$ in case it violates from below, and added to $V_3$ if violated from above.

Whenever a node from $V_1$ (or from $V_3$) violates its filter the algorithm terminates.
Additionally if (strictly) more than k nodes are in $V_1$ the algorithm is terminated or if (strictly) less than k nodes are in $V_1 \cup V_2$. 
If exactly k nodes are in $V_1$ and $n-k$ nodes are in $V_3$ the \Scattered is executed.

For the following argumentation on the competitiveness we focus on the case that \Scattered was not called since the analysis of \Scattered holds here.
Observe that OPT (with an error of $\varepsilon'$) had to communicate based on the following observation:

Let $t'$ be the time at which the algorithm terminates.
Assume to the contrary that OPT did not communicate during $[t, t']$.
In case node $i \in V_1$ observes a filter-violation from above, $\left(v_i < (1-\frac{\varepsilon}{2}) z \right)$ and $\varepsilon' \leq \frac{\varepsilon}{2}$, OPT had to set $\ell^* \leq v_i$ and $u^* \geq z$, which leads to a contradiction to the definition of filters. 
In case node $i \in V_3$ observes a filter-violation from below, $\left( v_i >\frac{1}{1-\varepsilon}(1-\frac{\varepsilon}{2})z \right)$ and $\varepsilon' \leq \frac{\varepsilon}{2}$, OPT had to set $u^* \geq v_i$ and $\ell^* \leq z$, which leads to a contradiction to the definition of filters.
The fact that OPT had to communicate in the remaining cases follows by the same arguments.
Since all cases lead to a contradiction, the bound on the competitiveness as stated above follows.
\end{proof}


\begin{thebibliography}{12}
\bibitem{chakrabarti}Arackaparambil, C., Brody, J., Chakrabarti, A.: Functional Monitoring without Monotonicity. In: Proceedings of the 36th International Colloquium on Automata, Languages and Programming, pp. 95--106. Springer, Berlin (2009)
%
\bibitem{cormodeSurvey}	Cormode, G.: The Continuous Distributed Monitoring Model. ACM SIGMOD Record 42.1, pp. 5--14. (2013)	
%
\bibitem{cormodeFunction} Cormode, G., Muthukrishnan, S., Ke, Y.: Algorithms for Distributed Functional Monitoring. ACM Transactions on Algorithms 7, 21 (2011)
%
\bibitem{giannakopoulos} Giannakopoulos Y., Koutsoupias, E.: Competitive Analysis of Maintaining Frequent Items of a Stream. Theoretical Computer Science 562, pp. 23--32. (2105)
%
\bibitem{lam} Lam, T.W., Liu, C.-M., Ting, H.-F.: Online Tracking of the Dominance Relationship of Distributed Multi-dimensional Data. In: Proceedings of the 8th International Workshop on Approximation and Online Algorithms, pp. 178--189. Springer, (2011)
%
\bibitem{mmm}M\"acker, A., Malatyali, M., Meyer auf der Heide, F.: Online Top-k-Position Monitoring of Distributed Data Streams. In: Proceedings of the 29th International Parallel and Distributed Processing Symposium, pp. 357--364. IEEE, (2015)
%
\bibitem{muthu} Muthukrishnan, S.: Data Streams: Algorithms and Applications. Now Publishers Inc, (2005)
%
\bibitem{phillips} Phillips, J., Verbin, E., Zhang, Q.: Lower Bounds for Number-in-Hand Multiparty Communication Complexity, Made Easy. In: Proceedings of the 23rd Annual ACM-SIAM Symposium on Discrete Algorithms, pp. 386--501. SIAM (2012)
%
\bibitem{sanders}Sanders, P., Schlag, S., M\"uller, I.: Communication Efficient Algorithms for Fundamental Big Data Problems. In: Proceedings of the IEEE International Conference on Big Data, pp. 15--23. IEEE, Silicon Valley (2013) 
%
\bibitem{tang} Tang M., Li F., Tao Y.: Distributed Online Tracking. In: Proceedings of the 2015 ACM SIGMOD International Conference on Management of Data, pp. 2047--2061. ACM, (2015)
%
\bibitem{yi} Yi, K., Zhang, Q.: Multidimensional Online Tracking. ACM Transactions on Algorithms 8, 12 (2012)
%
\bibitem{zhangFilter} Zhang, Z., Cheng, R., Papadias, D. and Tung, A.K.H.: Minimizing the Communication Cost for Continuous Skyline Maintenance. In: Proceedings of the ACM SIGMOD International Conference on Management of data, pp. 495--508. ACM, New York (2009)
%
\bibitem{zhangNoise} Zhang, Q.: Communication-Efficient Computation on Distributed Noisy Datasets. In: Proceedings of the 27th ACM Symposium on Parallelism in Algorithms and Architectures, pp. 313--322. ACM, (2015)
\end{thebibliography}
\end{document}